\documentclass[12pt]{article}
\usepackage[utf8]{inputenc}
\usepackage{geometry}
\usepackage{amsthm,amsmath,amssymb,amsfonts}
\usepackage{lmodern}
\usepackage{stmaryrd}
\usepackage{ifpdf}
\usepackage{lscape}
\usepackage{longtable}
\usepackage{microtype,mathtools}
\usepackage{tikz}
\usepackage{appendix,booktabs}
\newtheorem{theorem}{Theorem}
\newtheorem*{theorem*}{Theorem~2}
\newtheorem{lemma}[theorem]{Lemma}
\newtheorem{observation}[theorem]{Observation}
\newcommand{\FF}{\mathbb{F}}
\newcommand{\NN}{\mathbf{N}}
\newcommand{\RR}{\mathbf{R}}
\newcommand{\FR}{\mathcal{F}}
\newcommand{\AR}{\mathcal{A}}
\newcommand{\unlab}[2]{\left\llbracket #1\right\rrbracket_{#2}}
\newcommand{\abs}[1]{\left\lvert#1\right\rvert}
\newcommand{\littleo}[1]{o\mathopen{}\left(#1\right)}
\newcommand{\bigo}[1]{O\mathopen{}\left(#1\right)}
\ifpdf
\fi

\DeclareMathOperator{\im}{Im}
\renewcommand{\le}{\leqslant}
\renewcommand{\ge}{\geqslant}
\renewcommand{\leq}{\leqslant}
\renewcommand{\geq}{\geqslant}
\begin{document}
\title{A new bound for the 2/3 conjecture\thanks{This work was done in the framework of LEA STRUCO.}}
\author{Daniel Kr\'al'\thanks{Institute of Mathematics, DIMAP and Department of Computer Science, University of Warwick, Coventry CV4 7AL, United Kingdom. Previous affiliation: Institute of Computer Science (IUUK), Faculty of Mathematics and Physics, Charles University, Malostransk\'e n\'am\v est\'\i{} 25, 118 00 Prague 1, Czech Republic. E-mail: \texttt{D.Kral@warwick.ac.uk}. The work of this author leading to this invention has received funding from the European Research Council under the
European Union's Seventh Framework Programme (FP7/2007-2013)/ERC grant agreement no.~259385.}
\and
    Chun-Hung Liu\thanks{School of Mathematics, Georgia Institute of Technology,
    Atlanta, GA 30332, USA. E-mail: \texttt{cliu87@math.gatech.edu}.}
\and
	Jean-S{\'e}bastien Sereni\thanks{CNRS (LORIA), Nancy, France.
    E-mail: \texttt{sereni@kam.mff.cuni.cz}. This author's work was partially supported
    by the French \emph{Agence Nationale de la Recherche} under reference
    \textsc{anr 10 jcjc 0204 01}.}
\and
	Peter Whalen\thanks{School of Mathematics, Georgia Institute of Technology, Atlanta,
    GA 30332-0160, USA. E-mail: \texttt{pwhalen3@math.gatech.edu}.}
\and
	Zelealem B. Yilma\thanks{LIAFA (Universit\'e Denis Diderot), Paris, France.
    E-mail: \texttt{Zelealem.Yilma@liafa.jussieu.fr}. This author's work was
    supported by the French \emph{Agence Nationale de la Recherche} under reference
    \textsc{anr 10 jcjc 0204 01}.}
	}
\date{}
\maketitle
\begin{abstract}
We show that any $n$-vertex complete graph with edges colored with three colors
contains a set of at most four vertices such that the number of the neighbors of
these vertices in one of the colors is at least $2n/3$.
The previous best value, proved by Erd\H{o}s, Faudree, Gould, Gy\'arf\'as, Rousseau and Schelp
in $1989$, is $22$. It is conjectured that three vertices suffice.
\end{abstract}

\section{Introduction}
Erd\H{o}s and Hajnal~\cite{ErHa89}
made the observation that for a fixed positive integer $t$, a positive
real $\epsilon$, and a graph $G$ on $n > n_0$ vertices,
there is a set of $t$ vertices that have a neighborhood
of size at least $(1-(1+\epsilon)(2/3)^t)n$ in either $G$ or its complement.
They further inquired whether $2/3$ may be replaced by $1/2$.
This was answered in the affirmative by
Erd\H{o}s, Faudree, Gy\'arf\'as and Schelp~\cite{EFG+89},
who not only proved the result but also dispensed with the $(1+\epsilon)$
factor. They also phrased the question as a
problem of vertex domination in a multicolored graph.

Given a color $c$ in an $r$-coloring of the edges of
the complete graph, a subset $A$ of the vertex set
\emph{$c$-dominates} another subset $B$
if, for every $y \in B\setminus A$, there exists a vertex $x \in A$
such that the edge $xy$ is colored $c$.
The subset $A$ \emph{strongly} $c$-dominates $B$ if, in addition, for every
$y\in B\cap A$, there exists a vertex $x\in A$ such that $xy$ is colored $c$.
(Thus, the two notions coincide when $A\cap B=\emptyset$.)
The result of Erd\H{o}s \emph{et al.}~\cite{EFG+89}
may then be stated as follows.
\begin{theorem}
\label{thm-efgs}
For any fixed positive integer $t$
and any  $2$-coloring of the edges
of the complete graph $K_n$ on $n$ vertices,
there exist a color $c$ and a subset $X$ of size at most $t$ such that
all but at most $n/2^t$ vertices of $K_n$ are $c$-dominated by $X$.
\end{theorem}

In a more general form, they asked:
\emph{Given positive integers $r$, $t$, and $n$ along with an $r$-coloring of the edges of
the complete graph $K_n$ on $n$ vertices,
what is the largest subset $B$ of the vertices of $K_n$ necessarily
monochromatically dominated by some $t$-element subset of $K_n$?}
However, in the same paper~\cite{EFG+89},
 the authors presented a $3$-coloring of the edges of $K_n$ --- attributed to
 Kierstead --- which shows that if $r\ge3$, then it is not possible
 to monochromatically dominate all but a small fraction of the vertices with
 any fixed number $t$ of vertices.
 This $3$-coloring is defined as follows:
 the vertices of $K_n$ are partitioned into three sets $V_1, V_2, V_3$ of
 equal sizes
and an edge $xy$ with $x \in V_i$ and $y \in V_j$ is colored $i$ if $1 \leq i \leq j \leq 3$
and $j-i\le 1$ while edges between $V_1$ and $V_3$ are colored $3$.
Observe that, if $t$ is fixed, then at most $2n/3$ vertices may be monochromatically dominated.

In the other direction,
it was shown in the follow-up paper of
Erd\H{o}s, Faudree, Gould, Gy\'arf\'as, Rousseau and Schelp~\cite{EFG+90},
that if $t \geq 22$,
then, indeed, at least $2n/3$ vertices are monochromatically dominated in any
$3$-coloring of the edges of $K_n$.
The authors then
ask if $22$ may be replaced by a smaller number (specifically, $3$).
We prove here that $t\geq4$ is sufficient.
\begin{theorem}\label{thm:main}
For any $3$-coloring of the edges of $K_n$, where $n\ge2$,
there exist a color $c$ and a subset $A$ of at most four vertices
of $K_n$ such that $A$ strongly $c$-dominates at least $2n/3$ vertices of $K_n$.
\end{theorem}
In Kierstead's coloring, the number of colors appearing on the edges incident
with any given vertex is precisely $2$. As we shall see later on, this property
plays a central role in our arguments. In this regard,
our proof seems to suggest that Kierstead's coloring is somehow extremal,
giving more credence to the conjecture that
three vertices would suffice to monochromatically dominate a set of size $2n/3$
in any $3$-coloring of the edges of $K_n$.

We note that there exist $3$-colorings of the edges of $K_n$ such that no pair of vertices
monochromatically dominate $2n/3+\bigo{1}$ vertices. This can be seen by realizing that
in a random $3$-coloring, the probability that an arbitrary pair of vertices
monochromatically dominate more than $5n/9+\littleo{n}$ vertices is
$\littleo{1}$ by Chernoff's bound.

Our proof of Theorem~\ref{thm:main} utilizes the flag algebra theory introduced by Razborov,
which has recently led to numerous results in extremal graph and hypergraph theory.
In the following section, we present a brief introduction to the flag algebra
framework. The proof of Theorem~\ref{thm:main} is presented in
Section~\ref{main}.

We end this introduction by pointing out another interesting question:
what happens when one increases $r$, the number of colors?
Constructions in the vein of that of Kierstead --- for example,
partitioning $K_n$ into $s$ parts and using $r=\binom{s}{2}$
colors --- show that the size of dominated sets decreases with increasing $r$.
While it may be difficult to determine the minimum value of $t$ dominating
a certain proportion of the vertices,
it would be interesting to find out whether such constructions do, in fact,
give the correct bounds.

\section{Flag Algebras}
Flag algebras were introduced by Razborov~\cite{Raz07} as
a tool based on the graph limit theory of Lovász and Szegedy~\cite{LoSz06}
and Borgs \emph{et al.}~\cite{BCL+08} to approach
problems pertaining to extremal graph theory. This tool has been successfully
applied to various topics, such as
Turán-type problems~\cite{Raz10}, super-saturation questions~\cite{Raz08},
jumps in hypergraphs~\cite{BaTa11}, the Caccetta-Häggkvist
conjecture~\cite{HKN}, the chromatic number of common graphs~\cite{HHK+12} and the number
of pentagons in triangle-free graphs~\cite{Grz12,HHK+}. This list is far from
being exhaustive and results keep
coming~\cite{Bab,BaTa12,BHLL,CKP+,FaVaa,FaVab,HHN,Hir,KMS12,KMY13,Pik,PiRa}.

Let us now introduce the terminology related to flag algebras
needed in this paper. Since we deal with $3$-colorings of the
edges of complete graphs, we restrict our attention to this particular case.
Let us define a \emph{tricolored graph} to be a complete graph whose
edges are colored with $3$ colors.
If $G$ is a tricolored graph, then $V(G)$ is its vertex-set and
$\abs{G}$ is the number of vertices of $G$.
Let $\FF_\ell$ be the set of non-isomorphic tricolored graphs
with $\ell$ vertices, where two tricolored graphs are considered to be isomorphic
if they differ by a permutation of the vertices and a permutation of
the edge colors. (Therefore, which specific color is used for each edge is
irrelevant: what matters is whether or not pairs of edges are assigned the
same color.)
The elements of $\FF_3$ are shown in Figure \ref{fig-AB}.
We set $\FF\coloneqq\cup_{\ell\in\NN}\FF_\ell$.
Given a tricolored graph $\sigma$,
we define $\FF^{\sigma}_\ell$ to be the set of tricolored graphs $F$
on $\ell$ vertices with a fixed embedding of $\sigma$, that is,
an injective mapping $\nu$ from $V(\sigma)$ to $V(F)$ such that
$\im(\nu)$ induces in $F$ a subgraph
that differs from $\sigma$ only by a permutation of the edge colors.
The
elements of $\FF_{\ell}^{\sigma}$ are usually called \emph{$\sigma$-flags}
within the flag algebras framework.
We set $\FF^{\sigma}\coloneqq\cup_{\ell\in\NN}\FF^{\sigma}_\ell$.

\begin{figure}
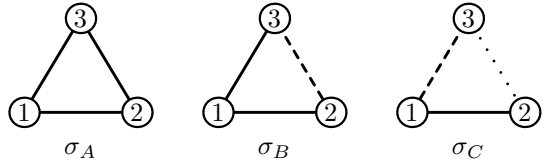

\begin{center}
\includegraphics{twothirds.11}
\hskip 5mm
\includegraphics{twothirds.12}
\hskip 5mm
\includegraphics{twothirds.13}
\end{center}
\caption{The elements of $\FF_3$.
         The edges of color $1$, $2$ and $3$ are represented by solid, dashed and dotted lines, respectively.}
\label{fig-AB}
\end{figure}

The central notions are factor algebras of $\FF$ and $\FF^{\sigma}$ equipped
with addition and multiplication. Let us start with the simpler case of $\FF$.
If $H\in\FF$ and $H'\in\FF_{\abs{H}+1}$, then
$p(H,H')$ is the probability that a randomly chosen subset of
$\abs{H}$ vertices of $H'$ induces a subgraph isomorphic to $H$.
For a set $F$, we define $\RR F$ to be the set of all formal
linear combinations of elements of $F$ with real coefficients.
Let $\AR\coloneqq\RR\FF$ and
let $\FR$ be $\AR$ factorised
by the subspace of $\RR\FF$ generated by all combinations of
the form
\[H-\sum_{H'\in\FF_{\abs{H}+1}}p(H,H')H'.\]

Next, we define the multiplication on $\AR$ based on the elements of $\FF$ as
follows.
If $H_1$ and $H_2$ are two elements of $\FF$ and
$H\in\FF_{\abs{H_1}+\abs{H_2}}$, then $p(H_1,H_2;H)$ is the probability that
two randomly chosen disjoint subsets of vertices of $H$ with sizes $\abs{H_1}$
and $\abs{H_2}$ induce subgraphs isomorphic to $H_1$ and $H_2$, respectively.
We set
\[H_1\cdot H_2\coloneqq\sum_{H\in\FF_{\abs{H_1}+\abs{H_2}}}p(H_1,H_2;H)H.\]
The multiplication is linearly extended to $\RR\FF$. Standard elementary
probability computations~\cite[Lemma 2.4]{Raz07} show that this multiplication in $\RR\FF$ gives rise
to a well-defined multiplication in the factor algebra $\AR$.

The definition of $\AR^\sigma$ follows the same lines.
Let $H$ and $H'$ be two tricolored graphs in $\FF^\sigma$
with embeddings $\nu$ and $\nu'$ of $\sigma$. Informally, we consider
the copy of $\sigma$ in $H'$ and we extend it into an element of
$\FF_{\abs{H}}^{\sigma}$ by randomly choosing additional vertices in
$H'$. We are interested in the probability that this random extension is
isomorphic to $H$ and the isomorphism preserves the embeddings of $\sigma$.
Formally, we let $p(H,H')$ be the probability that $\nu'(V(\sigma))$ together with
a randomly chosen subset of $\abs{H}-\abs{\sigma}$ vertices in
$V(H')\setminus\nu'(V(\sigma))$ induce a subgraph that
is isomorphic
to $H$ through an isomorphism $f$ that preserves the embeddings,
that is, $\nu'=f\circ\nu$.
The set $\AR^\sigma$ is composed of all formal real linear combinations of elements
of $\RR\FF^\sigma$ factorised by the subspace of $\RR\FF^\sigma$ generated by all combinations of
the form
\[H-\sum_{H'\in\FF^\sigma_{\abs{H}+1}}p(H,H')H'.\]
Similarly, $p(H_1,H_2;H)$ is the probability that $\nu(V(\sigma))$ together with
two randomly chosen disjoint subsets of $\abs{H_1}-\abs{\sigma}$ and
$\abs{H_2}-\abs{\sigma}$ vertices in $V(H)\setminus\nu(V(\sigma))$
induce subgraphs isomorphic to $H_1$ and $H_2$, respectively,
with the isomorphisms preserving the embeddings of $\sigma$.
The definition of the product is then analogous to that in $\AR$.

Consider an infinite sequence $(G_i)_{i\in\NN}$ of tricolored
graphs with an increasing number of vertices. Recall that if $H\in\FF$,
then $p(H,G_i)$ is the probability that a randomly chosen subset of $\abs{H}$
vertices of $G_i$ induces a subgraph isomorphic to $H$.
The sequence $(G_i)_{i\in\NN}$ is \emph{convergent} if $p(H,G_i)$ has a limit
for every $H\in\FF$.
A standard argument (using Tychonoff's theorem~\cite{Tyc30}) yields that every infinite sequence of tricolored
graphs has a convergent (infinite) subsequence.

The results presented in this and the next paragraph were established by Razborov~\cite{Raz07}.
Fix now a convergent sequence $(G_i)_{i\in\NN}$ of tricolored graphs.
We set $q(H)\coloneqq\lim_{i\to\infty} p(H,G_i)$ for every $H\in\FF$,
and we linearly extend $q$ to $\AR$.
The obtained mapping $q$ is a homomorphism from $\AR$ to $\RR$.
Moreover, for $\sigma\in\FF$ and an embedding $\nu$ of $\sigma$ in $G_i$,
define $p^\nu_i(H)\coloneqq p(H,G_i)$.
Picking $\nu$ at random thus gives rise to a random distribution of mappings
from $\AR^{\sigma}$ to $\RR$, for each $i\in\NN$.
Since $p(H,G_i)$ converges (as $i$ tends to infinity) for every $H\in\FF$,
the sequence of these distributions must also converge.
In fact, $q$ itself fully determines
the random distributions of $q^\sigma$ for all $\sigma$.
In what follows,
$q^\sigma$ will be a randomly chosen mapping from $\AR^{\sigma}$ to $\RR$
based on the limit distribution. Any mapping $q^\sigma$ from support
of the limit distribution is a homomorphism from $\AR^{\sigma}$ to $\RR$.

Let us now have a closer look at the relation between $q$ and $q^\sigma$.
The ``averaging'' operator
$\unlab{\cdot}{\sigma}\colon \AR^{\sigma} \to \AR$ is a linear operator
defined on the elements of $\FF^\sigma$ by
$\unlab{H}{\sigma}\coloneqq p\cdot H'$, where
$H'$ is the (unlabeled) tricolored graph in $\FF$ corresponding to $H$ and
$p$ is the probability that a random injective mapping
from $V(\sigma)$ to $V(H')$ is an embedding of $\sigma$ in $H'$ yielding $H$.
The key relation between $q$ and $q^\sigma$ is the following:
\begin{equation}
\label{eq-extave}
\forall H\in\AR^\sigma,\quad q(\unlab{H}{\sigma})=\int q^\sigma(H),
\end{equation}
where the integration is over the probability space given
by the limit random distribution of $q^\sigma$.
We immediately conclude that
if $q^\sigma(H)\ge 0$ almost surely,
then $q(\unlab{H}{\sigma})\ge 0$.
In particular,
\begin{equation}\label{eq:square}
\forall H\in\AR^\sigma,\quad q(\unlab{H^2}{\sigma})\ge 0.
\end{equation}

\subsection{Particular Notation Used in our Proof}
\label{notation}
Before presenting the proof of Theorem~\ref{thm:main}, we need to introduce some notation and
several lemmas.
Recall that $\sigma_A$, $\sigma_B$ and $\sigma_C$,
the elements of $\FF_3$, are given in Figure~\ref{fig-AB}.
For $i\in\{A,B,C\}$ and a triple $t\in\{1,2,3\}^3$, let $F^{i}_t$ be
the element of $\FF_4^{\sigma_i}$ in which the unlabeled vertex of $F^i_t$
is joined by an edge of color $t_j$ to the image of the $j$-th vertex of $\sigma_i$ for $j \in \{1,2,3\}$.
Two elements of $\AR^{\sigma_B}$ and two of $\AR^{\sigma_C}$ will be of
interest in our further considerations:
\begin{align*}
w_B &\coloneqq 165 F^B_{113} + 165 F^B_{333} -279 F^B_{123} -44 F^B_{131} +328
F^B_{133} +10 F^B_{233} + 421 F^B_{323},\\
w'_B &\coloneqq -580 F^B_{113} - 580 F^B_{333} +668 F^B_{123} -264 F^B_{131} +10
F^B_{133} +725 F^B_{233} + 632 F^B_{323},\\
w_C &\coloneqq 100 F^C_{112} + 100 F^C_{312} -100 F^C_{113} -100 F^C_{133} +162
F^C_{122} +163 F^C_{221},\quad\text{and}\\
w'_C &\coloneqq -10 F^C_{112} - 10 F^C_{312} +10 F^C_{113} +10 F^C_{133} -77
F^C_{122} +89 F^C_{221}.
\end{align*}

We make use of seven elements $\sigma_1,\ldots,\sigma_7$ out of the $15$
elements of $\FF_4$. They are depicted
in Figure~\ref{fig-cores}.
For $i\in\{1,\ldots,7\}$ and a quadruple $d\in\{1,2,3\}^4$, let $F^{i}_d$
be the element of $\FF_5^{\sigma_i}$ such that the unlabeled vertex of $F^i_d$
is joined by an edge of color $d_j$ to the $j$-th vertex of $\sigma_i$ for $j \in \{1,2,3,4\}$.
If $i \in \{1,\ldots,7\}$ and $c\in\{1,2,3\}$, then
$F^i_{(c)}$ is the element of $\AR^{\sigma_i}$ that is the sum of all the five-vertex
$\sigma_i$-flags $F^i_d$ such that the unlabeled vertex is joined by an edge of color $c$
to at least one of the vertices of $\sigma_i$, i.e., at least one of the entries of $d$
is $c$.

\begin{figure}
\begin{center}
\includegraphics{twothirds.1}
\includegraphics{twothirds.2}
\includegraphics{twothirds.3}
\includegraphics{twothirds.4}
\includegraphics{twothirds.5}
\includegraphics{twothirds.6}
\includegraphics{twothirds.7}
\end{center}
\caption{The elements $\sigma_1,\ldots,\sigma_7$ of $\FF_4$.
         The edges of color $1$, $2$ and $3$ are represented by solid, dashed and dotted lines, respectively.}
\label{fig-cores}
\end{figure}

Finally, we define $H_1,\ldots,H_{142}$ to be the elements of $\FF_5$
in the way depicted in Appendix~\ref{flagfigs}.

\section{Proof of Theorem~\ref{thm:main}}
\label{main}
In this section, we prove Theorem~\ref{thm:main} by contradiction: in a series
of lemmas, we shall prove some properties of a counterexample which
eventually allow us to establish the nonexistence of counterexamples.
Specifically, we first find a number of flag inequalities by hand and then
we combine them with appropriate coefficients to obtain a contradiction. The
coefficients are found with the help of a computer.

Let $G$ be a tricolored complete graph. For a vertex $v$ of $G$,
let $A_v$ be the set of colors of the edges incident with $v$.
Consider a sequence of graphs $(G_k)_{k\in\NN}$, obtained from $G$ by replacing
each vertex $v$ of $G$ with a complete graph of order $k$
with edges colored uniformly at random with colors in $A_v$;
the colors of the edges between the complete graphs corresponding to the vertices
$v$ and $v'$ of $G$ are assigned the color of the edge $vv'$.
This sequence of graphs converges asymptotically almost surely;
let $q_G$ be the corresponding homomorphism from $\AR$ to $\RR$.

Let $n \geq 2$.
We define a \emph{counterexample} to be a tricolored graph with $n$ vertices
such that for every color $c\in\{1,2,3\}$,
each set $W$ of at most four vertices strongly $c$-dominates less than $2n/3$ vertices of $G$.
A counterexample readily satisfies the following property.
\begin{observation}
\label{obs-mintwo}
If $G$ is a counterexample, then every vertex is incident with edges of at least two different colors.
\end{observation}

In the next lemma, we establish an inequality that $q_G$ satisfies if $G$ is a counterexample.
To do so, define the quantity $\varepsilon_c(\sigma_i)$ for $i \in \{1,\ldots,7\}$ and $c\in\{1,2,3\}$
to be $1/2$ if $\sigma_i$ contains a single edge with color $c$,
$-1/3$ if each vertex of $\sigma_i$ is incident with an edge colored $c$,
$1/6$ if $\sigma_i$ contains at least two edges with color $c$ and a vertex incident with edges of a single color different from $c$, and
$0$, otherwise. These values are gathered in Table~\ref{tb-eps}. Let us
underline that, unlike in most of the previous applications of flag algebras,
we do need to deal with second-order terms (specifically, $\bigo{1/n}$ terms)
in our flag inequalities to establish Theorem~\ref{thm:main}.
\begin{table}
\begin{center}
\begin{tabular}{cccccccc}
\toprule
& i=1 & i=2 & i=3 & i=4 & i=5 & i=6 & i=7 \\
\midrule
$c=1$ & -1/3 & 0 & -1/3 & -1/3 & 0 & 0 & 0 \\
$c=2$ & 1/2 & 0 & 1/6 & -1/3 & -1/3 & -1/3 & 0 \\
$c=3$ & 1/2 & 1/2 & 1/2 & 1/2 & 1/2 & 0 & 0 \\
\bottomrule
\end{tabular}
\end{center}
\caption{The values $\varepsilon_c(\sigma_i)$ for $i \in \{1,\ldots,7\}$ and $c\in\{1,2,3\}$.}
\label{tb-eps}
\end{table}

\begin{lemma}
\label{lm-avg0}
Let $G$ be a counterexample with $n$ vertices.
For every $i \in \{1,\ldots,7\}$ and $c\in\{1,2,3\}$, a homomorphism $q_G^{\sigma_i}$
from $\AR^{\sigma_i}$ to $\RR$ almost surely satisfies the inequality
\[q_G^{\sigma_i}(F_{(c)}^i)\le \frac{2}{3}+\frac{\varepsilon_c(\sigma_i)}{n}.\]
\end{lemma}

\begin{proof}
Fix $i\in\{1,\ldots,7\}$ and $c\in\{1,2,3\}$.
Consider the graph $G_k$ for sufficiently large $k$.
Let $(w_1, w_2, w_3, w_4)$ be a randomly selected quadruple of vertices of $G_k$
inducing a subgraph isomorphic to $\sigma_i$. Further, let $W$ be the
set of vertices strongly $c$-dominated by $\{w_1,\ldots,w_4\}$. We show that
$\abs{W}\le \frac{2nk}{3}+\varepsilon_c(\sigma_i)k+o(k)$ with probability
tending to one as $k$ tends to infinity.
This will establish the inequality stated in the lemma. Indeed, it implies
that for every $\eta>0$, there exists $k_\eta$ such that
if $k>k_\eta$, then
$q_{G_k}^{\sigma_i}(F_{(c)}^i)\le\frac{2}{3}+\frac{\varepsilon_c(\sigma_i)}{n}+\eta$
with probability at least $1-\eta$.
As $q_{G_k}^{\sigma_i}(F_{(c)}^i)$ tends to $q_{G}^{\sigma_i}(F_{(c)}^i)$ as
$k$ tends to infinity, we obtain the stated inequality with probability $1$.

For $i \in \{1,2,3,4\}$,
let $v_i$ be the vertex of $G$ corresponding to the clique $W_i$ of $G_k$
containing $w_i$.
Let $V$ be the set of vertices of $G$ that are strongly $c$-dominated by $\{v_1,\ldots,v_4\}$.
Since $G$ is a counterexample, $\abs{V}<2n/3$,
and hence, $\abs{V}\le 2n/3-1/3$.
If $w_j$ and $w_{j'}$ are joined by an edge of color $c$ and, furthermore,
$v_j=v_{j'}$,
then $v_j$ is added to $V$ as well.
Since $V$ is still strongly $c$-dominated
by a quadruple of vertices in $G$ (replace $v_{j'}$ by any of its $c$-neighbors),
it follows that $\abs{V}\le 2n/3-1/3$.

The set $W$ can contain the $\abs{V}k$
vertices of the cliques corresponding to the vertices in $V$, and,
potentially,
it also contains some additional vertices if $w_i$ has no
$c$-neighbors among $w_1,\ldots,w_4$. In this case, the additional
vertices in $W$ are the $c$-neighbors of $w_i$ in $W_i$. With high
probability, there are at most $k/3+o(k)$ such vertices if $v_i$ is
incident with edges of all three colors in $G$, and at most $k/2+o(k)$
if $v_i$ is incident with edges of only two colors in $G$.

If $\varepsilon_c(\sigma_i)=-1/3$, then all the vertices $w_1,\ldots,w_4$
have a $c$-neighbor among $w_1,\ldots,w_4$ and thus $W$ contains only vertices
of the cliques corresponding to the vertices $V$. We conclude that
$\abs{W}\le \frac{(2n-1)k}{3}+o(k)$, as required.

If $\varepsilon_c(\sigma_i)=0$, then all but one of the vertices $w_1,\ldots,w_4$
have a $c$-neighbor among $w_1,\ldots,w_4$ and the vertex $w_j$
that has none is incident in $\sigma_i$ with edges of the two colors different
from $c$. In particular, either $w_j$ has no $c$-neighbors inside $W_j$ or
$v_j$ is incident with edges of three distinct colors in $G$.
This implies that $\abs{W}\le\frac{(2n-1)k}{3}+o(k)$ in the former case and
$\abs{W}\le\frac{2nk}{3}+o(k)$ in the latter case. So, the bound holds.

If $\varepsilon_c(\sigma_i)=1/6$, then all but one of the vertices among $w_1,\ldots,w_4$
have a $c$-neighbor among $w_1,\ldots,w_4$. Let $w_j$ be the exceptional vertex.
Since $w_j$ has at most $k/2+o(k)$ $c$-neighbors in $W_j$,
it follows that $\abs{W}\le\frac{2nk}{3}+\frac{k}{6}+o(k)$.

Finally, if $\varepsilon_c(\sigma_i)=1/2$, then two vertices
$w_j$ and $w_{j'}$ among $w_1,\ldots,w_4$ have
no $c$-neighbors in $\{w_1,\ldots,w_4\}$. The vertices $w_j$ and $w_{j'}$
have at most $k/2+o(k)$ $c$-neighbors each in $W_j$ and $W_{j'}$,
respectively.
Moreover, since
$\sigma_i$ contains edges of all three colors, one of $w_j$ and $w_{j'}$
is incident in $\sigma_i$ with edges of the two colors different from $c$.
Hence, this vertex has at most $k/3+o(k)$ $c$-neighbors in $W_j$.
We conclude that the set $W$ contains at most $\abs{V}k+5k/6+o(k)\le\frac{2nk}{3}+\frac{k}{2}+o(k)$
vertices.
\end{proof}

As a consequence of~(\ref{eq-extave}),
we have the following corollary of Lemma~\ref{lm-avg0}.
\begin{lemma}
\label{lm-avg}
Let $G$ be a counterexample with $n$ vertices.
For every $i \in \{1,\ldots,7\}$ and $c\in\{1,2,3\}$ such that $\varepsilon_c(\sigma_i)\le 0$,
it holds that
\[q_G(\unlab{2\sigma_i/3-F_{(c)}^i}{\sigma_i})\ge 0.\]
\end{lemma}

We now prove that in a counterexample, at most two colors are
used to color the edges incident with any given vertex. As we shall see, this structural
property of counterexamples directly implies their nonexistence, thereby
proving Theorem~\ref{thm:main}.
\begin{lemma}
\label{lm-claw}
No counterexample contains a vertex incident with edges of all three colors.
\end{lemma}
\begin{proof}
Let $G$ be a counterexample and $w_3\in\RR\FF_5$ be the sum of all elements of $\FF_5$ that contain
a vertex incident with at least three colors. By the definition of $q_G$,
the graph $G$ has a vertex incident with edges of all three colors if and only if
$q_G(w_3)>0$. Lemma~\ref{lm-avg} implies that
$q_G(H)$ is non-negative for each element $H$ of $\AR$ corresponding to any
column of Table~\ref{tb-claw1} (in Appendix~\ref{clawtable}). In addition,
\eqref{eq:square} ensures that $q_G(H)$ is also non-negative for each element
$H$ of $\AR$ corresponding to any of the first four columns of
Table~\ref{tb-claw2} (in Appendix~\ref{clawtable}). Note that these elements
can be expressed as elements of $\RR\FF_5$.
Summing these columns with coefficients
\[
\begin{array}{ccccccc}
\frac{23457815885978657985}{1029505785512512},&&
\frac{134730108347752975}{4596007971038},&&
\frac{134730108347752975}{4596007971038},\\
\frac{15852088219609163945}{514752892756256},&&
\frac{196791037567187109905}{12354069426150144},&&
\frac{33245823856447882025}{24708138852300288},\\
\frac{3956624143678293415}{772129339134384},&&
\frac{30762195734543710715}{772129339134384},&&
\frac{20816545085118359705}{4118023142050048},\\
\frac{74313622711306287405}{2059011571025024},&&
\frac{48968798259015}{514752892756256},&&
\frac{39315342699665}{6177034713075072},\\
\frac{15977347300925119}{32944185136400384},&&
\frac{8880723226482731}{24708138852300288},&&
\end{array}
\]
respectively, yields an element
$w_0$ of $\AR$ given in the very last column of Table~\ref{tb-claw2}. Notice
that for every $H\in\FF_5$, the coefficient of $H$ in $-w_0$ is at least the
coefficient of $H$ in $w_3$. In particular, the sum $w_3+w_0$, which
belongs to $\RR\FF_5$, has only non-positive coefficients. We now view both
$w_0$ and $w_3$ as elements of $\AR$ and use that
$q_G$ is a homomorphism from $\AR$ to $\RR$. First of all, $q_G(w_3+w_0)\le0$.
So, we derive that $q_G(w_3)\le-q_G(w_0)$. As noted earlier, $q_G(H)\ge0$ for each
element $H$ used to define $w_0$.
Hence, since none of the above (displayed) coefficients is negative,
we deduce that $q_G(w_0)\ge0$. Consequently, $q_G(w_3)\le 0$,
which therefore implies that $q_G(w_3)=0$. This means that $G$ has no vertex
incident with edges of all three colors.
\end{proof}

We are now in a position to prove Theorem~\ref{thm:main}, whose statement
is recalled below.
\begin{theorem*}
\label{thm-main}
Let $n\ge2$.
Every tricolored graph with $n$ vertices contains a subset of at most four vertices
that strongly $c$-dominates at least $2n/3$ vertices for some color $c$.
\end{theorem*}
\begin{proof}
Suppose, on the contrary, that there exists a counterexample $G$.
Recall that $A_v$ is the set of colors that appear on the edges incident to
the vertex $v$.
Now, by Observation~\ref{obs-mintwo} and Lemma~\ref{lm-claw},
it holds that $\abs{A_v}=2$ for every vertex $v$ of $G$.
Hence, $V(G)$ can be partitioned into three sets
$V_1$, $V_2$ and $V_3$, where $v \in V_i$ if and only if $i \notin A_v$.
Without loss of generality,
assume that $\abs{V_1} \geq \abs{V_2} \geq \abs{V_3}$.
Pick $u \in V_1$ and $v \in V_2$.
As $A_u \cap A_w = \{3\}$ for all $w \in V_2$,
we observe that $V_2$ is $3$-dominated by $\{u\}$.
Similarly, $V_1$ is $3$-dominated by $\{v\}$.
Therefore, the set $\{u,v\}$ strongly $3$-dominates
$V_1 \cup V_2$, which has size at least $2n/3$.
\end{proof}

\section{Concluding remarks}
It is natural to ask what bound can be proven for domination with three vertices.
Here, it does not seem that the trick we used in this paper helps.
We can prove only that
every tricolored graph with $n$ vertices contains a subset of at most three vertices
that $c$-dominates at least $0.66117n$ vertices for some color $c$.

We believe the difficulty we face is caused by the following phenomenon.
The average number of vertices dominated by a triple isomorphic
to $\sigma_A$ or $\sigma_B$ (see Figure~\ref{fig-AB} for notation)
is bounded away from $2/3$ in the graphs $(G_k)_{k\in\NN}$, which are
described at the beginning of Section~\ref{main}, for $G$ being the rainbow triangle.
So, if any of these two configurations is used, a tight bound cannot be proven
since the inequalities analogous to that in Lemma~\ref{lm-avg} are not tight and
no triple of vertices dominates more than $2/3$ of the vertices in $(G_k)_{k\in\NN}$
to compensate this deficiency.

We see that if we aimed to prove a tight result, we can only average over rainbow
triangles (which are isomorphic to $\sigma_C$).
Now consider the following graph $G$: start from the disjoint union of a large clique
of order $2m$ with all edges colored $1$ and a rainbow triangle. For
$i\in\{1,2\}$, join exactly $m$ vertices of the clique to all three vertices of the rainbow triangle
by edges colored $i$. The obtained simple complete graph has exactly one
rainbow triangle, which dominates about half of the vertices.
Thus, the average proportion of vertices dominated by triples isomorphic to $\sigma_C$
in the graphs $(G_k)_{k\in\NN}$ is close to $1/2$.
This phenomenon does not occur for quadruples of vertices.

\bibliographystyle{siam}
\bibliography{flags}

\newpage
\pagestyle{empty}
\appendix
\section[Appendix]{The Elements of $\FF_5$}
\label{flagfigs}

\begin{center}
\includegraphics[width=15mm]{twothirds.101}
\includegraphics[width=15mm]{twothirds.102}
\includegraphics[width=15mm]{twothirds.103}
\includegraphics[width=15mm]{twothirds.104}
\includegraphics[width=15mm]{twothirds.105}
\includegraphics[width=15mm]{twothirds.106}
\includegraphics[width=15mm]{twothirds.107}
\includegraphics[width=15mm]{twothirds.108}
\includegraphics[width=15mm]{twothirds.109}
\includegraphics[width=15mm]{twothirds.110}
\includegraphics[width=15mm]{twothirds.111}
\includegraphics[width=15mm]{twothirds.112}
\includegraphics[width=15mm]{twothirds.113}
\includegraphics[width=15mm]{twothirds.114}
\includegraphics[width=15mm]{twothirds.115}
\includegraphics[width=15mm]{twothirds.116}
\includegraphics[width=15mm]{twothirds.117}
\includegraphics[width=15mm]{twothirds.118}
\includegraphics[width=15mm]{twothirds.119}
\includegraphics[width=15mm]{twothirds.120}
\includegraphics[width=15mm]{twothirds.121}
\includegraphics[width=15mm]{twothirds.122}
\includegraphics[width=15mm]{twothirds.123}
\includegraphics[width=15mm]{twothirds.124}
\includegraphics[width=15mm]{twothirds.125}
\includegraphics[width=15mm]{twothirds.126}
\includegraphics[width=15mm]{twothirds.127}
\includegraphics[width=15mm]{twothirds.128}
\includegraphics[width=15mm]{twothirds.129}
\includegraphics[width=15mm]{twothirds.130}
\includegraphics[width=15mm]{twothirds.131}
\includegraphics[width=15mm]{twothirds.132}
\includegraphics[width=15mm]{twothirds.133}
\includegraphics[width=15mm]{twothirds.134}
\includegraphics[width=15mm]{twothirds.135}
\includegraphics[width=15mm]{twothirds.136}
\includegraphics[width=15mm]{twothirds.137}
\includegraphics[width=15mm]{twothirds.138}
\includegraphics[width=15mm]{twothirds.139}
\includegraphics[width=15mm]{twothirds.140}
\includegraphics[width=15mm]{twothirds.141}
\includegraphics[width=15mm]{twothirds.142}
\includegraphics[width=15mm]{twothirds.143}
\includegraphics[width=15mm]{twothirds.144}
\includegraphics[width=15mm]{twothirds.145}
\includegraphics[width=15mm]{twothirds.146}
\includegraphics[width=15mm]{twothirds.147}
\includegraphics[width=15mm]{twothirds.148}
\includegraphics[width=15mm]{twothirds.149}
\includegraphics[width=15mm]{twothirds.150}
\includegraphics[width=15mm]{twothirds.151}
\includegraphics[width=15mm]{twothirds.152}
\includegraphics[width=15mm]{twothirds.153}
\includegraphics[width=15mm]{twothirds.154}
\includegraphics[width=15mm]{twothirds.155}
\includegraphics[width=15mm]{twothirds.156}
\includegraphics[width=15mm]{twothirds.157}
\includegraphics[width=15mm]{twothirds.158}
\includegraphics[width=15mm]{twothirds.159}
\includegraphics[width=15mm]{twothirds.160}
\includegraphics[width=15mm]{twothirds.161}
\includegraphics[width=15mm]{twothirds.162}
\includegraphics[width=15mm]{twothirds.163}
\includegraphics[width=15mm]{twothirds.164}
\includegraphics[width=15mm]{twothirds.165}
\includegraphics[width=15mm]{twothirds.166}
\includegraphics[width=15mm]{twothirds.167}
\includegraphics[width=15mm]{twothirds.168}
\includegraphics[width=15mm]{twothirds.169}
\includegraphics[width=15mm]{twothirds.170}
\includegraphics[width=15mm]{twothirds.171}
\includegraphics[width=15mm]{twothirds.172}
\includegraphics[width=15mm]{twothirds.173}
\includegraphics[width=15mm]{twothirds.174}
\includegraphics[width=15mm]{twothirds.175}
\includegraphics[width=15mm]{twothirds.176}
\includegraphics[width=15mm]{twothirds.177}
\includegraphics[width=15mm]{twothirds.178}
\includegraphics[width=15mm]{twothirds.179}
\includegraphics[width=15mm]{twothirds.180}
\includegraphics[width=15mm]{twothirds.181}
\includegraphics[width=15mm]{twothirds.182}
\includegraphics[width=15mm]{twothirds.183}
\includegraphics[width=15mm]{twothirds.184}
\includegraphics[width=15mm]{twothirds.185}
\includegraphics[width=15mm]{twothirds.186}
\includegraphics[width=15mm]{twothirds.187}
\includegraphics[width=15mm]{twothirds.188}
\includegraphics[width=15mm]{twothirds.189}
\includegraphics[width=15mm]{twothirds.190}
\includegraphics[width=15mm]{twothirds.191}
\includegraphics[width=15mm]{twothirds.192}
\includegraphics[width=15mm]{twothirds.193}
\includegraphics[width=15mm]{twothirds.194}
\includegraphics[width=15mm]{twothirds.195}
\includegraphics[width=15mm]{twothirds.196}
\includegraphics[width=15mm]{twothirds.197}
\includegraphics[width=15mm]{twothirds.198}
\includegraphics[width=15mm]{twothirds.199}
\includegraphics[width=15mm]{twothirds.200}
\includegraphics[width=15mm]{twothirds.201}
\includegraphics[width=15mm]{twothirds.202}
\includegraphics[width=15mm]{twothirds.203}
\includegraphics[width=15mm]{twothirds.204}
\includegraphics[width=15mm]{twothirds.205}
\includegraphics[width=15mm]{twothirds.206}
\includegraphics[width=15mm]{twothirds.207}
\includegraphics[width=15mm]{twothirds.208}
\includegraphics[width=15mm]{twothirds.209}
\includegraphics[width=15mm]{twothirds.210}
\includegraphics[width=15mm]{twothirds.211}
\includegraphics[width=15mm]{twothirds.212}
\includegraphics[width=15mm]{twothirds.213}
\includegraphics[width=15mm]{twothirds.214}
\includegraphics[width=15mm]{twothirds.215}
\includegraphics[width=15mm]{twothirds.216}
\includegraphics[width=15mm]{twothirds.217}
\includegraphics[width=15mm]{twothirds.218}
\includegraphics[width=15mm]{twothirds.219}
\includegraphics[width=15mm]{twothirds.220}
\includegraphics[width=15mm]{twothirds.221}
\includegraphics[width=15mm]{twothirds.222}
\includegraphics[width=15mm]{twothirds.223}
\includegraphics[width=15mm]{twothirds.224}
\includegraphics[width=15mm]{twothirds.225}
\includegraphics[width=15mm]{twothirds.226}
\includegraphics[width=15mm]{twothirds.227}
\includegraphics[width=15mm]{twothirds.228}
\includegraphics[width=15mm]{twothirds.229}
\includegraphics[width=15mm]{twothirds.230}
\includegraphics[width=15mm]{twothirds.231}
\includegraphics[width=15mm]{twothirds.232}
\includegraphics[width=15mm]{twothirds.233}
\includegraphics[width=15mm]{twothirds.234}
\includegraphics[width=15mm]{twothirds.235}
\includegraphics[width=15mm]{twothirds.236}
\includegraphics[width=15mm]{twothirds.237}
\includegraphics[width=15mm]{twothirds.238}
\includegraphics[width=15mm]{twothirds.239}
\includegraphics[width=15mm]{twothirds.240}
\includegraphics[width=15mm]{twothirds.241}
\includegraphics[width=15mm]{twothirds.242}
\end{center}

\newpage
\section{Vectors Used in the Proof of Lemma~\ref{lm-claw}}
\label{clawtable}
\small
\begin{center}
\begin{longtable}{ccccccccccc}
\caption{The first ten vectors}\label{tb-claw1}\\
\toprule
&
\begin{tikzpicture}
	\node[rotate=90] {$\unlab{2\sigma_1/3-F_{(1)}^1}{\sigma_1}$};
\end{tikzpicture} &
\begin{tikzpicture}
 	\node[rotate=90] {$\unlab{2\sigma_2/3-F_{(1)}^2}{\sigma_2}$};
 \end{tikzpicture} &
 \begin{tikzpicture}
	\node[rotate=90] {$\unlab{2\sigma_2/3-F_{(2)}^2}{\sigma_2}$};
\end{tikzpicture} &
\begin{tikzpicture}
	\node[rotate=90] {$\unlab{2\sigma_3/3-F_{(1)}^3}{\sigma_3}$};
\end{tikzpicture} &
\begin{tikzpicture}
	\node[rotate=90] {$\unlab{2\sigma_4/3-F_{(1)}^4}{\sigma_4}$};
\end{tikzpicture} &
\begin{tikzpicture}
	\node[rotate=90] {$\unlab{2\sigma_4/3-F_{(2)}^4}{\sigma_4}$};
\end{tikzpicture} &
 \begin{tikzpicture}
	 \node[rotate=90] {$\unlab{2\sigma_5/3-F_{(1)}^5}{\sigma_5}$};
 \end{tikzpicture}  &
 \begin{tikzpicture}
	 \node[rotate=90] {$\unlab{2\sigma_5/3-F_{(2)}^5}{\sigma_5}$};
 \end{tikzpicture} &
 \begin{tikzpicture}
	 \node[rotate=90] {$\unlab{2\sigma_6/3-F_{(1)}^6}{\sigma_6}$};
 \end{tikzpicture} &
 \begin{tikzpicture}
	 \node[rotate=90] {$\unlab{2\sigma_7/3-F_{(2)}^7}{\sigma_7}$};
 \end{tikzpicture} \\
 \midrule
 \endfirsthead
\multicolumn{11}{c}%
{\tablename\ \thetable\ -- \textit{Continued from previous page}} \\
\toprule
&
\begin{tikzpicture}
	\node[rotate=90] {$\unlab{2\sigma_1/3-F_{(1)}^1}{\sigma_1}$};
\end{tikzpicture} &
\begin{tikzpicture}
 	\node[rotate=90] {$\unlab{2\sigma_2/3-F_{(1)}^2}{\sigma_2}$};
 \end{tikzpicture} &
 \begin{tikzpicture}
	\node[rotate=90] {$\unlab{2\sigma_2/3-F_{(2)}^2}{\sigma_2}$};
\end{tikzpicture} &
\begin{tikzpicture}
	\node[rotate=90] {$\unlab{2\sigma_3/3-F_{(1)}^3}{\sigma_3}$};
\end{tikzpicture} &
\begin{tikzpicture}
	\node[rotate=90] {$\unlab{2\sigma_4/3-F_{(1)}^4}{\sigma_4}$};
\end{tikzpicture} &
\begin{tikzpicture}
	\node[rotate=90] {$\unlab{2\sigma_4/3-F_{(2)}^4}{\sigma_4}$};
\end{tikzpicture} &
 \begin{tikzpicture}
	 \node[rotate=90] {$\unlab{2\sigma_5/3-F_{(1)}^5}{\sigma_5}$};
 \end{tikzpicture}  &
 \begin{tikzpicture}
	 \node[rotate=90] {$\unlab{2\sigma_5/3-F_{(2)}^5}{\sigma_5}$};
 \end{tikzpicture} &
 \begin{tikzpicture}
	 \node[rotate=90] {$\unlab{2\sigma_6/3-F_{(1)}^6}{\sigma_6}$};
 \end{tikzpicture} &
 \begin{tikzpicture}
	 \node[rotate=90] {$\unlab{2\sigma_7/3-F_{(2)}^7}{\sigma_7}$};
 \end{tikzpicture} \\
\midrule
\endhead
\bottomrule \multicolumn{11}{r}{\textit{Continued on next page}} \\
\endfoot
\bottomrule
\endlastfoot
$H_{1}$ & 	0 &	0 &	0 &	 0 &	0 &	0 &	0 &	0 &	 0 &	0 \\
$H_{2}$ &	  0 &	 0 &	 0 &	 0 &	 0 &	 0 &	 0 &	 0 &	 0 &	 0 \\
$H_{3}$ & 	 0 &	 0 &	 0 &	 0 &	 0 &	 0 &	 0 &	 0 &	 0 &	 0 \\
$H_{4}$ & 	 -1/90 &	 0 &	 0 &	 0 &	 0 &	 0 &	 0 &	 0 &	 0 &	 0 \\
$H_{5}$ & 0 &	 0 &	 0 &	 0 &	 0 &	 0 &	 0 &	 0 &	 0 &	 0 \\
$H_{6}$ &	 -1/90 &	 -1/180 &	 1/90 &	 0 &	 0 &	 0 &	 0 &	 0 &	 0 &	 0 \\
$H_{7}$ & 0 &	 0 &	 0 &	 0 &	 0 &	 0 &	 0 &	 0 &	 0 &	 0 \\
$H_{8}$ &	 0 &	 -1/60 &	 -1/60 &	 0 &	 0 &	 0 &	 0 &	 0 &	 0 &	 0 \\
$H_{9}$ 	& 0 &	 -1/45 &	 2/45 &	 0 &	 0 &	 0 &	 0 &	 0 &	 0 &	 0 \\
$H_{10}$ 	 &0 &	 0 &	 0 &	 0 &	 0 &	 0 &	 0 &	 0 &	 0 &	 0 \\
$H_{11}$ 	 &0 &	 0 &	 0 &	 -1/90 &	 0 &	 0 &	 0 &	 0 &	 0 &	 0 \\
$H_{12}$ 	& 0 &	 0 &	 0 &	 0 &	 0 &	 0 &	 0 &	 0 &	 0 &	 0 \\
$H_{13}$ 	 &0 &	 0 &	 0 &	 0 &	 0 &	 0 &	 0 &	 0 &	 0 &	 0 \\
$H_{14}$ 	 &-1/90 &	 0 &	 0 &	 0 &	 -1/180 &	 1/90 &	 0 &	 0 &	 0 &	 0 \\
$H_{15}$ 	 &0 &	 0 &	 0 &	 0 &	 0 &	 0 &	 0 &	 0 &	 0 &	 0 \\
$H_{16}$ 	 &-1/180 &	 0 &	 0 &	 -1/180 &	 0 &	 0 &	 -1/360 &	 1/180 &	 0 &	 0 \\
$H_{17}$ 	 &-1/90 &	 0 &	 0 &	 -1/180 &	 0 &	 0 &	 0 &	 0 &	 1/180 &	 0 \\
$H_{18}$ 	& 0 &	 0 &	 0 &	 0 &	 0 &	 0 &	 0 &	 0 &	 0 &	 0 \\
$H_{19}$ 	& -1/180 &	 0 &	 0 &	 0 &	 0 &	 0 &	 1/180 &	 -1/360 &	 0 &	 0 \\
$H_{20}$ 	& -1/90 &	 1/90 &	 -1/180 &	 0 &	 0 &	 0 &	 0 &	 0 &	 0 &	 0 \\
$H_{21}$ 	& -1/180 &	 0 &	 0 &	 -1/180 &	 0 &	 0 &	 0 &	 0 &	 0 &	 1/120 \\
$H_{22}$ 	& 0 &	 0 &	 0 &	 1/180 &	 0 &	 0 &	 0 &	 0 &	 0 &	 0 \\
$H_{23}$ 	& 0 &	 0 &	 0 &	 0 &	 0 &	 0 &	 0 &	 0 &	 0 &	 0 \\
$H_{24}$ 	& 0 &	 0 &	 0 &	 0 &	 0 &	 0 &	 0 &	 0 &	 0 &	 0 \\
$H_{25}$ 	& 0 &	 -1/180 &	 -1/180 &	 0 &	 -1/90 &	 -1/90 &	 0 &	 0 &	 0 &	 0 \\
$H_{26}$ 	& 0 &	 0 &	 0 &	 0 &	 0 &	 0 &	 0 &	 0 &	 0 &	 0 \\
$H_{27}$ 	& 0 &	 -1/180 &	 -1/180 &	 0 &	 0 &	 0 &	 -1/180 &	 -1/180 &	 0 &	 0 \\
$H_{28}$ 	& 0 &	 -1/90 &	 -1/90 &	 0 &	 0 &	 0 &	 0 &	 0 &	 1/90 &	 0 \\
$H_{29}$ 	& -1/90 &	 0 &	 0 &	 0 &	 0 &	 0 &	 0 &	 0 &	 0 &	 0 \\
$H_{30}$ 	& -1/180 &	 -1/180 &	 -1/180 &	 0 &	 0 &	 0 &	 -1/360 &	 -1/360 &	 0 &	 0 \\
$H_{31}$ 	& -1/180 &	 -1/180 &	 1/90 &	 0 &	 -1/180 &	 1/90 &	 1/180 &	 -1/360 &	 0 &	 0 \\
$H_{32}$ 	& 0 &	 -1/60 &	 -1/60 &	 0 &	 0 &	 0 &	 0 &	 0 &	 0 &	 0 \\
$H_{33}$ 	& 0 &	 -1/90 &	 1/180 &	 0 &	 0 &	 0 &	 -1/360 &	 1/180 &	 0 &	 0 \\
$H_{34}$ 	& 0 &	 -1/90 &	 -1/90 &	 1/90 &	 0 &	 0 &	 0 &	 0 &	 -1/90 &	 0 \\
$H_{35}$ 	& 0 &	 0 &	 0 &	 0 &	 0 &	 0 &	 0 &	 0 &	 0 &	 0 \\
$H_{36}$ 	& 0 &	 -1/180 &	 1/90 &	 0 &	 0 &	 0 &	 1/90 &	 -1/180 &	 0 &	 0 \\
$H_{37}$ 	& 0 &	 0 &	 0 &	 0 &	 0 &	 0 &	 -1/180 &	 -1/180 &	 0 &	 0 \\
$H_{38}$ 	& 0 &	 1/90 &	 1/90 &	 0 &	 0 &	 0 &	 0 &	 0 &	 0 &	 0 \\
$H_{39}$ 	& 0 &	 -1/180 &	 -1/180 &	 0 &	 0 &	 0 &	 0 &	 0 &	 0 &	 0 \\
$H_{40}$ 	& 0 &	 0 &	 0 &	 -1/90 &	 0 &	 0 &	 0 &	 0 &	 0 &	 0 \\
$H_{41}$ 	& 0 &	 0 &	 0 &	 0 &	 0 &	 0 &	 0 &	 0 &	 0 &	 0 \\
$H_{42}$ 	& 0 &	 0 &	 0 &	 -1/90 &	 0 &	 0 &	 0 &	 0 &	 0 &	 0 \\
$H_{43}$ 	& -1/90 &	 0 &	 0 &	 0 &	 1/90 &	 -1/180 &	 0 &	 0 &	 0 &	 0 \\
$H_{44}$ 	 &1/180 &	 0 &	 0 &	 -1/180 &	 0 &	 0 &	 0 &	 0 &	 0 &	 0 \\
$H_{45}$ 	& -1/180 &	 0 &	 0 &	 -1/90 &	 0 &	 0 &	 1/180 &	 1/180 &	 0 &	 0 \\
$H_{46}$ 	& -1/90 &	 0 &	 0 &	 0 &	 0 &	 0 &	 0 &	 0 &	 1/45 &	 0 \\
$H_{47}$ 	& 0 &	 0 &	 0 &	 0 &	 0 &	 0 &	 0 &	 0 &	 0 &	 0 \\
$H_{48}$ 	& 0 &	 0 &	 0 &	 0 &	 0 &	 0 &	 0 &	 0 &	 0 &	 0 \\
$H_{49}$ 	& 0 &	 0 &	 0 &	 0 &	 0 &	 0 &	 0 &	 0 &	 0 &	 0 \\
$H_{50}$ 	& 0 &	 0 &	 0 &	 0 &	 -1/180 &	 -1/180 &	 -1/180 &	 -1/180 &	 0 &	 0 \\
$H_{51}$ 	& -1/180 &	 0 &	 0 &	 0 &	 -1/180 &	 -1/180 &	 -1/360 &	 -1/360 &	 0 &	 0 \\
$H_{52}$ 	& -1/180 &	 0 &	 0 &	 -1/180 &	 0 &	 0 &	 -1/360 &	 -1/360 &	 0 &	 0 \\
$H_{53}$ 	& -1/180 &	 0 &	 0 &	 0 &	 -1/180 &	 -1/180 &	 0 &	 0 &	 1/180 &	 0 \\
$H_{54}$ 	& 0 &	 0 &	 0 &	 0 &	 0 &	 0 &	 0 &	 0 &	 0 &	 0 \\
$H_{55}$ 	& 0 &	 0 &	 0 &	 0 &	 0 &	 0 &	 -1/180 &	 -1/180 &	 0 &	 0 \\
$H_{56}$ 	& 0 &	 -1/180 &	 -1/180 &	 0 &	 -1/90 &	 -1/90 &	 0 &	 0 &	 0 &	 0 \\
$H_{57}$ 	& -1/180 &	 -1/180 &	 -1/180 &	 0 &	 0 &	 0 &	 -1/360 &	 -1/360 &	 0 &	 0 \\
$H_{58}$ 	& 0 &	 -1/180 &	 -1/180 &	 0 &	 0 &	 0 &	 1/360 &	 -1/180 &	 1/180 &	 0 \\
$H_{59}$ 	& 0 &	 -1/180 &	 1/90 &	 0 &	 -1/90 &	 -1/90 &	 0 &	 0 &	 1/180 &	 0 \\
$H_{60}$ 	& 0 &	 -1/180 &	 1/90 &	 -1/180 &	 0 &	 0 &	 -1/180 &	 -1/180 &	 0 &	 0 \\
$H_{61}$ 	& 0 &	 -1/180 &	 -1/180 &	 1/90 &	 0 &	 0 &	 0 &	 0 &	 -1/45 &	 0 \\
$H_{62}$ 	& -1/180 &	 0 &	 0 &	 0 &	 0 &	 0 &	 -1/360 &	 -1/360 &	 0 &	 0 \\
$H_{63}$ 	& -1/180 &	 0 &	 0 &	 0 &	 0 &	 0 &	 1/360 &	 1/360 &	 0 &	 0 \\
$H_{64}$ 	& 0 &	 0 &	 0 &	 0 &	 0 &	 0 &	 1/90 &	 -1/180 &	 0 &	 0 \\
$H_{65}$ 	& 0 &	 0 &	 0 &	 1/90 &	 0 &	 0 &	 -1/360 &	 -1/360 &	 0 &	 0 \\
$H_{66}$ 	& 0 &	 -1/180 &	 -1/180 &	 0 &	 0 &	 0 &	 -1/180 &	 -1/180 &	 0 &	 0 \\
$H_{67}$ 	& 0 &	 -1/180 &	 -1/180 &	 0 &	 0 &	 0 &	 0 &	 0 &	 0 &	 0 \\
$H_{68}$ 	& 0 &	 -1/180 &	 1/90 &	 0 &	 1/90 &	 -1/180 &	 1/360 &	 -1/180 &	 0 &	 0 \\
$H_{69}$ 	& 1/90 &	 -1/90 &	 1/180 &	 0 &	 0 &	 0 &	 0 &	 0 &	 0 &	 0 \\
$H_{70}$ 	& 0 &	 -1/180 &	 -1/180 &	 1/90 &	 0 &	 0 &	 -1/360 &	 1/180 &	 0 &	 -1/120 \\
$H_{71}$ 	& 0 &	 0 &	 0 &	 0 &	 0 &	 0 &	 -1/180 &	 -1/180 &	 1/180 &	 0 \\
$H_{72}$ 	& 0 &	 0 &	 0 &	 1/90 &	 0 &	 0 &	 0 &	 0 &	 0 &	 -1/60 \\
$H_{73}$ 	& 0 &	 0 &	 0 &	 -1/60 &	 0 &	 0 &	 0 &	 0 &	 0 &	 0 \\
$H_{74}$ 	& -1/90 &	 1/90 &	 1/90 &	 -1/180 &	 0 &	 0 &	 0 &	 0 &	 0 &	 0 \\
$H_{75}$ 	& -1/45 &	 0 &	 0 &	 0 &	 0 &	 0 &	 0 &	 0 &	 0 &	 0 \\
$H_{76}$ 	& -1/90 &	 0 &	 0 &	 -1/90 &	 1/90 &	 1/90 &	 0 &	 0 &	 0 &	 0 \\
$H_{77}$ 	& -1/90 &	 0 &	 0 &	 0 &	 0 &	 0 &	 1/90 &	 -1/180 &	 0 &	 0 \\
$H_{78}$ 	& -1/90 &	 0 &	 0 &	 0 &	 0 &	 0 &	 0 &	 0 &	 0 &	 0 \\
$H_{79}$ 	& -1/90 &	 1/45 &	 -1/90 &	 0 &	 -1/180 &	 1/90 &	 0 &	 0 &	 0 &	 0 \\
$H_{80}$ 	& -1/180 &	 0 &	 0 &	 0 &	 -1/180 &	 -1/180 &	 -1/360 &	 -1/360 &	 0 &	 0 \\
$H_{81}$ 	& 1/90 &	 0 &	 0 &	 0 &	 0 &	 0 &	 1/360 &	 -1/180 &	 0 &	 0 \\
$H_{82}$ 	& -1/180 &	 -1/180 &	 -1/180 &	 0 &	 0 &	 0 &	 1/180 &	 -1/360 &	 1/180 &	 0 \\
$H_{83}$ 	& 1/90 &	 1/60 &	 -1/60 &	 0 &	 0 &	 0 &	 -1/360 &	 1/180 &	 0 &	 0 \\
$H_{84}$ 	& -1/90 &	 0 &	 0 &	 0 &	 -1/90 &	 1/45 &	 0 &	 0 &	 0 &	 0 \\
$H_{85}$ 	& -1/90 &	 0 &	 0 &	 0 &	 0 &	 0 &	 -1/180 &	 1/360 &	 -1/90 &	 0 \\
$H_{86}$ 	& -1/180 &	 -1/180 &	 -1/180 &	 0 &	 -1/180 &	 -1/180 &	 -1/360 &	 -1/360 &	 0 &	 0 \\
$H_{87}$ 	& 1/90 &	 -1/180 &	 -1/180 &	 -1/180 &	 -1/180 &	 1/90 &	 -1/360 &	 -1/360 &	 0 &	 0 \\
$H_{88}$ 	& 1/90 &	 -1/180 &	 -1/180 &	 0 &	 1/90 &	 -1/180 &	 0 &	 0 &	 -1/90 &	 -1/120 \\
$H_{89}$ 	& -1/180 &	 -1/180 &	 -1/180 &	 0 &	 1/180 &	 1/180 &	 -1/360 &	 -1/360 &	 0 &	 0 \\
$H_{90}$ 	& 1/90 &	 -1/180 &	 -1/180 &	 0 &	 0 &	 0 &	 -1/180 &	 1/90 &	 0 &	 -1/120 \\
$H_{91}$ 	& 1/45 &	 -1/90 &	 -1/90 &	 0 &	 0 &	 0 &	 0 &	 0 &	 -1/90 &	 0 \\
$H_{92}$ 	& 0 &	 0 &	 0 &	 0 &	 0 &	 0 &	 0 &	 0 &	 0 &	 0 \\
$H_{93}$ 	& 0 &	 0 &	 0 &	 0 &	 0 &	 0 &	 0 &	 0 &	 0 &	 0 \\
$H_{94}$ 	& 0 &	 -1/180 &	 -1/180 &	 0 &	 -1/90 &	 -1/90 &	 0 &	 0 &	 0 &	 0 \\
$H_{95}$ 	& 0 &	 1/90 &	 -1/180 &	 0 &	 0 &	 0 &	 1/90 &	 -1/180 &	 0 &	 0 \\
$H_{96}$ 	& 0 &	 -1/180 &	 1/90 &	 0 &	 -1/90 &	 -1/90 &	 0 &	 0 &	 1/180 &	 0 \\
$H_{97}$ 	& 0 &	 1/90 &	 -1/180 &	 1/90 &	 0 &	 0 &	 -1/180 &	 -1/180 &	 0 &	 0 \\
$H_{98}$ 	& 0 &	 0 &	 0 &	 0 &	 0 &	 0 &	 0 &	 0 &	 -1/30 &	 0 \\
$H_{99}$ 	& 0 &	 0 &	 0 &	 -1/60 &	 0 &	 0 &	 0 &	 0 &	 0 &	 0 \\
$H_{100}$ &	 0 &	 -1/90 &	 -1/90 &	 0 &	 0 &	 0 &	 0 &	 0 &	 1/90 &	 0 \\
$H_{101}$ &	 0 &	 2/45 &	 -1/45 &	 0 &	 0 &	 0 &	 0 &	 0 &	 0 &	 0 \\
$H_{102}$ &	 0 &	 -1/90 &	 -1/90 &	 0 &	 -1/90 &	 -1/90 &	 0 &	 0 &	 0 &	 0 \\
$H_{103}$ &	 0 &	 1/45 &	 -1/90 &	 0 &	 1/90 &	 -1/180 &	 -1/180 &	 -1/180 &	 0 &	 0 \\
$H_{104}$ &	 0 &	 -1/90 &	 -1/90 &	 0 &	 1/45 &	 -1/90 &	 0 &	 0 &	 -1/90 &	 0 \\
$H_{105}$ &	 0 &	 -1/90 &	 -1/90 &	 0 &	 0 &	 0 &	 0 &	 0 &	 -1/90 &	 0 \\
$H_{106}$ &	 0 &	 1/45 &	 -1/90 &	 -1/180 &	 -1/90 &	 -1/90 &	 0 &	 0 &	 0 &	 0 \\
$H_{107}$ &	 0 &	 0 &	 0 &	 -1/45 &	 0 &	 0 &	 0 &	 0 &	 0 &	 0 \\
$H_{108}$ &	 0 &	 0 &	 0 &	 -1/45 &	 0 &	 0 &	 0 &	 0 &	 0 &	 0 \\
$H_{109}$ &	 0 &	 0 &	 0 &	 -1/90 &	 0 &	 0 &	 0 &	 0 &	 0 &	 0 \\
$H_{110}$ &	 0 &	 0 &	 0 &	 -1/180 &	 0 &	 0 &	 -1/180 &	 -1/180 &	 0 &	 0 \\
$H_{111}$ &	 0 &	 0 &	 0 &	 -1/180 &	 -1/180 &	 -1/180 &	 1/180 &	 -1/360 &	 0 &	 0 \\
$H_{112}$ &	 0 &	 0 &	 0 &	 -1/90 &	 0 &	 0 &	 0 &	 0 &	 0 &	 0 \\
$H_{113}$ &	 0 &	 0 &	 0 &	 -1/180 &	 -1/180 &	 -1/180 &	 1/180 &	 -1/360 &	 0 &	 0 \\
$H_{114}$ &	 0 &	 0 &	 0 &	 -1/90 &	 0 &	 0 &	 1/360 &	 -1/180 &	 0 &	 0 \\
$H_{115}$ &	 0 &	 0 &	 0 &	 -1/180 &	 -1/90 &	 -1/90 &	 0 &	 0 &	 -1/90 &	 0 \\
$H_{116}$ &	 0 &	 0 &	 0 &	 -1/180 &	 -1/90 &	 1/180 &	 -1/180 &	 -1/180 &	 0 &	 0 \\
$H_{117}$ &	 0 &	 0 &	 0 &	 -1/180 &	 -1/180 &	 -1/180 &	 -1/180 &	 1/360 &	 -1/90 &	 0 \\
$H_{118}$ &	 0 &	 0 &	 0 &	 1/180 &	 0 &	 0 &	 -1/180 &	 -1/180 &	 0 &	 0 \\
$H_{119}$ &	 0 &	 0 &	 0 &	 1/90 &	 0 &	 0 &	 -1/360 &	 -1/360 &	 -1/90 &	 -1/120 \\
$H_{120}$ &	 0 &	 0 &	 0 &	 -1/90 &	 0 &	 0 &	 0 &	 0 &	 0 &	 0 \\
$H_{121}$ &	 0 &	 0 &	 0 &	 -1/180 &	 -1/180 &	 1/90 &	 -1/180 &	 1/360 &	 0 &	 0 \\
$H_{122}$ &	 0 &	 0 &	 0 &	 -1/180 &	 0 &	 0 &	 -1/360 &	 -1/360 &	 0 &	 0 \\
$H_{123}$ &	 0 &	 0 &	 0 &	 -1/180 &	 0 &	 0 &	 -1/180 &	 -1/180 &	 -1/90 &	 0 \\
$H_{124}$ &	 0 &	 0 &	 0 &	 -1/180 &	 0 &	 0 &	 -1/90 &	 1/180 &	 0 &	 0 \\
$H_{125}$ &	 0 &	 0 &	 0 &	 1/180 &	 0 &	 0 &	 -1/360 &	 -1/360 &	 0 &	 -1/120 \\
$H_{126}$ &	 0 &	 0 &	 0 &	 -1/180 &	 0 &	 0 &	 -1/180 &	 -1/180 &	 1/180 &	 0 \\
$H_{127}$ &	 0 &	 0 &	 0 &	 -1/180 &	 1/180 &	 -1/90 &	 -1/360 &	 -1/360 &	 0 &	 -1/120 \\
$H_{128}$ &	 0 &	 0 &	 0 &	 1/90 &	 0 &	 0 &	 0 &	 0 &	 -1/90 &	 -1/60 \\
$H_{129}$ &	 0 &	 0 &	 0 &	 0 &	 0 &	 0 &	 -1/180 &	 -1/180 &	 0 &	 0 \\
$H_{130}$ &	 0 &	 0 &	 0 &	 0 &	 0 &	 0 &	 0 &	 0 &	 0 &	 0 \\
$H_{131}$ &	 0 &	 0 &	 0 &	 0 &	 -1/180 &	 -1/180 &	 -1/180 &	 -1/180 &	 0 &	 0 \\
$H_{132}$ &	 0 &	 0 &	 0 &	 0 &	 0 &	 0 &	 1/180 &	 -1/90 &	 0 &	 0 \\
$H_{133}$ &	 0 &	 0 &	 0 &	 0 &	 -1/180 &	 -1/180 &	 -1/360 &	 -1/360 &	 1/180 &	 0 \\
$H_{134}$ &	 0 &	 0 &	 0 &	 0 &	 -1/60 &	 0 &	 -1/180 &	 -1/180 &	 0 &	 0 \\
$H_{135}$ &	 0 &	 0 &	 0 &	 0 &	 0 &	 -1/60 &	 0 &	 0 &	 0 &	 -1/60 \\
$H_{136}$ &	 0 &	 0 &	 0 &	 0 &	 -1/90 &	 -1/90 &	 -1/180 &	 -1/180 &	 -1/90 &	 0 \\
$H_{137}$ &	 0 &	 0 &	 0 &	 0 &	 1/90 &	 -1/180 &	 -1/120 &	 0 &	 0 &	 -1/120 \\
$H_{138}$ &	 0 &	 0 &	 0 &	 0 &	 -1/180 &	 -1/180 &	 -1/180 &	 1/360 &	 -1/45 &	 0 \\
$H_{139}$ &	 0 &	 0 &	 0 &	 0 &	 -1/180 &	 -1/180 &	 -1/360 &	 -1/360 &	 -1/90 &	 -1/120 \\
$H_{140}$ &	 0 &	 0 &	 0 &	 0 &	 0 &	 0 &	 -1/180 &	 1/90 &	 0 &	 -1/60 \\
$H_{141}$ &	 0 &	 0 &	 0 &	 0 &	 0 &	 0 &	 -1/90 &	 -1/90 &	 0 &	 0 \\
$H_{142}$ &	0 &	0 &	0 &	0 &	0 &	0 &	0 &	0 &	 -2/45 & 	0 \\
\end{longtable}

\newpage
\begin{longtable}{ccccc|cc}
\caption{The last six vectors}\label{tb-claw2}\\
\toprule
&	 $\unlab{w_B\cdot w_B}{\sigma_B}$
 &	$\unlab{w'_B\cdot w'_B}{\sigma_B}$
&	 $\unlab{w_C\cdot w_C}{\sigma_C}$
 &	$\unlab{w'_C\cdot w'_C}{\sigma_C}$
&	 $w_3$
&	 $w_0$ \\
\midrule
 \endfirsthead
\multicolumn{7}{c}%
{\tablename\ \thetable\ -- \textit{Continued from previous page}} \\
\toprule
&	 $\unlab{w_B\cdot w_B}{\sigma_B}$
 &	$\unlab{w'_B\cdot w'_B}{\sigma_B}$
&	 $\unlab{w_C\cdot w_C}{\sigma_C}$
 &	$\unlab{w'_C\cdot w'_C}{\sigma_C}$
&	 $w_3$
&	 $w_0$ \\
\midrule
\endhead
\bottomrule \multicolumn{7}{r}{\textit{Continued on next page}} \\
\endfoot
\bottomrule
\endlastfoot
$H_{1}$ 	&	 0 &	0 &	0 &	 0 &	 0 &	 0 	\\
$H_{2}$ 	&	 0 &	 0 &	 0 &	 0 &	 0 &	 0 	\\
$H_{3}$ 	&	 0 &	 0 &	 0 &	 0 &	 0 &	 0 	\\
$H_{4}$ 	&	 0 &	 0 &	 0 &	 0 &	 1 &	$\frac{ -1563854392398577199}{6177034713075072}$ 	 \\
$H_{5}$ 	&	 0 &	 0 &	 0 &	 0 &	 0 &	 0 \\
$H_{6}$ 	&	 29161/60 &	 101524/15 &	 0 &	 0 &	 1 &	 -1 	\\
$H_{7}$ 	&	 0 &	 0 &	 0 &	 0 &	 0 &	 0 	\\
$H_{8}$ 	&	 0 &	 0 &	 2000 &	 20 &	 0 &	 0	\\
$H_{9}$ 	&	 0 &	 0 &	 -4000/3 &	 -40/3 &	 0 &	 0 	\\
$H_{10}$ 	&	 0 &	 0 &	 0 &	 0 &	 0 &	 0 	\\
$H_{11}$ 	&	 1815/2 &	 33640/3 &	 0 &	 0 &	 1 &	 $\frac{ -10173977739002723}{55152095652456 	}$ \\
$H_{12}$ 	&	 0 &	 0 &	 0 &	 0 &	 0 &	 0 	\\
$H_{13}$ 	&	 0 &	 0 &	 0 &	 0 &	 0 &	 0 	\\
$H_{14}$ 	&	 -242 &	 5104 &	 0 &	 0 &	 1 &	 $\frac{ -734882450141728337}{2316388017403152 	}$ \\
$H_{15}$ 	&	 0 &	 0 &	 0 &	 0 &	 0 &	 0 	\\
$H_{16}$ 	&	 -9922/15 &	 -422/3 &	 0 &	 0 &	 1 &	 $\frac{ -5722046702587908817}{37062208278450432 	}$ \\
$H_{17}$ 	&	 57013/60 &	 -65634/5 &	 0 &	 0 &	 1 &	 $\frac{ -57717650068438077139}{148248833113801728	}$ \\
$H_{18}$ 	&	 0 &	 0 &	 0 &	 0 &	 0 &	 0 \\
$H_{19}$ 	&	 0 &	 0 &	 0 &	 0 &	 1 &	 $\frac{ -703462682135213465}{3369291661677312 	}$ \\
$H_{20}$ 	&	 29161/60 &	 101524/15 &	 0 &	 0 &	 1 &	 -1 	\\
$H_{21}$ 	&	 781/2 &	 -37700/3 &	 0 &	 0 &	 1 &	 $\frac{ -5891700664190917297}{148248833113801728 	}$ \\
$H_{22}$ 	&	 -1804 &	 -580/3 &	 0 &	 0 &	 1 &	$\frac{  -32614888977443071}{18531104139225216 	}$ \\
$H_{23}$ 	&	 0 &	 0 &	 0 &	 0 &	 0 &	 0 	\\
$H_{24}$ 	&	 0 &	 0 &	 0 &	 0 &	 0 &	 0 	\\
$H_{25}$ 	&	 0 &	 0 &	 19723/20 &	 2047/20 &	 1 &	$\frac{  -15461491234942018543}{5929953324552069120 	}$ \\
$H_{26}$ 	&	 0 &	 0 &	 0 &	 0 &	 0 &	 0 	\\
$H_{27}$ 	&	 0 &	 0 &	 540 &	 77/3 &	 1 &	 $\frac{ -88140807390257339}{289548502175394 	}$ \\
$H_{28}$ 	&	 10/3 &	 105125/6 &	 0 &	 0 &	 1 &	 $\frac{ -35834405989042100849}{74124416556900864 }$ \\
$H_{29}$ 	&	 0 &	 0 &	 0 &	 0 &	 1 &	$\frac{  -1563854392398577199}{6177034713075072 }$ \\
$H_{30}$ 	&	 0 &	 0 &	 270 &	 77/6 &	 1 &	 $\frac{ -5456161234717178191}{12354069426150144 	}$ \\
$H_{31}$ 	&	 0 &	 0 &	 0 &	 0 &	 1 &	$\frac{  -4427934211353668633}{37062208278450432 	}$ \\
$H_{32}$ 	&	 0 &	 0 &	 2000 &	 20 &	 0 &	 0 	\\
$H_{33}$ 	&	 0 &	 0 &	 -1810/3 &	 -97/6 &	 1 &	 $\frac{ -1570031427111652271}{6177034713075072 	}$ \\
$H_{34}$ 	&	 -328/3 &	 725/3 &	 2000/3 &	 20/3 &	 1 &	 $\frac{ -327323049775204219}{6738583323354624 	}$ \\
$H_{35}$ 	&	 0 &	 0 &	 0 &	 0 &	 0 &	 0 	\\
$H_{36}$ 	&	 0 &	 0 &	 0 &	 0 &	 1 &	 $\frac{ -2040849950139277}{1323650295658944 	}$ \\
$H_{37}$ 	&	 0 &	 0 &	 2187/5 &	 5929/60 &	 1 &	 $\frac{ -324486989357699}{150414806324880 	}$ \\
$H_{38}$ 	&	 0 &	 0 &	 -4000/3 &	 -40/3 &	 0 &	 0 	\\
$H_{39}$ 	&	 177241/60 &	 99856/15 &	 0 &	 0 &	 1 &	 $\frac{ -14389006173379021}{6177034713075072 }$ \\
$H_{40}$ 	&	 53792/15 &	 10/3 &	 0 &	 0 &	 1 &	 -1 \\
$H_{41}$ 	&	 0 &	 0 &	 0 &	 0 &	 0 &	 0 	\\
$H_{42}$ 	&	 53792/15 &	 10/3 &	 0 &	 0 &	 1 &	 -1 	\\
$H_{43}$ 	&	 -242 &	 5104 &	 0 &	 0 &	 1 &	 $\frac{ -2444189262506217731}{32944185136400384 	}$ \\
$H_{44}$ 	&	 -19723/60 &	 8018 &	 0 &	 0 &	 1 &	$\frac{  -130980504818216225}{5294601182635776 	}$ \\
$H_{45}$ 	&	 19251/20 &	 -46426/5 &	 0 &	 0 &	 1 &	$\frac{  -767941410949255531}{4118023142050048 	}$ \\
$H_{46}$ 	&	 0 &	 0 &	 0 &	 0 &	 1 &	 $\frac{ -5219817337367791253}{37062208278450432 	}$ \\
$H_{47}$ 	&	 -4743/20 &	 -27388/5 &	 0 &	 0 &	 1 &	 $\frac{ -354709127257189891}{6177034713075072 	}$ \\
$H_{48}$ 	&	 0 &	 0 &	 0 &	 0 &	 0 &	 0 	\\
$H_{49}$ 	&	 0 &	 0 &	 0 &	 0 &	 0 &	 0 	\\
$H_{50}$ 	&	 0 &	 0 &	 0 &	 0 &	 1 &	 $\frac{ -102522009006261748933}{296497666227603456 	}$ \\
$H_{51}$ 	&	 0 &	 0 &	 4401/10 &	 -6853/60 &	 1 &	 $\frac{ -103816701767414115797}{592995332455206912 	}$ \\
$H_{52}$ 	&	 1331/4 &	 24476/3 &	 0 &	 0 &	 1 &	 $\frac{ -1794830760611264087}{5294601182635776	}$ \\
$H_{53}$ 	&	 -7157/30 &	 72838/15 &	 0 &	 0 &	 1 &	 $\frac{ -55226415700070668835}{296497666227603456 	}$ \\
$H_{54}$ 	&	 0 &	 0 &	 0 &	 0 &	 0 &	 0 	\\
$H_{55}$ 	&	 0 &	 0 &	 2187/5 &	 5929/60 &	 1 &	 $\frac{ -324486989357699}{150414806324880 	}$ \\
$H_{56}$ 	&	 0 &	 0 &	 1630/3 &	 -89/3 &	 1 &	 $\frac{ -148706888944854103}{561548610279552 	}$ \\
$H_{57}$ 	&	 0 &	 0 &	 270 &	 77/6 &	 1 &	 $\frac{ -5456161234717178191}{12354069426150144 	}$ \\
$H_{58}$ 	&	 0 &	 0 &	 0 &	 0 &	 1 &	 $\frac{ -74826771055029195907}{148248833113801728 	}$ \\
$H_{59}$ 	&	 0 &	 0 &	 0 &	 0 &	 1 &	 -1 	\\
$H_{60}$ 	&	 0 &	 0 &	 -540 &	 -77/3 &	 1 &	 $\frac{ -127346913837154513}{240663690119808 	}$ \\
$H_{61}$ 	&	 -93 &	 48430/3 &	 0 &	 0 &	 1 &	$\frac{  -10466732649688555}{5294601182635776 	}$ \\
$H_{62}$ 	&	 0 &	 0 &	 0 &	 0 &	 1 &	 $\frac{ -9320739160958665481}{37062208278450432 	}$ \\
$H_{63}$ 	&	 0 &	 0 &	 0 &	 0 &	 1 &	$\frac{  -62387193432797713}{37062208278450432 	}$ \\
$H_{64}$ 	&	 0 &	 0 &	 0 &	 0 &	 1 &	$\frac{  -217609023306544037}{1323650295658944 	}$ \\
$H_{65}$ 	&	 -34522/15 &	 316/3 &	 0 &	 0 &	 1 &	 -1 	\\
$H_{66}$ 	&	 0 &	 0 &	 540 &	 77/3 &	 1 &	 $\frac{ -88140807390257339}{289548502175394 	}$ \\
$H_{67}$ 	&	 177241/60 &	 99856/15 &	 0 &	 0 &	 1 &	$\frac{  -14389006173379021}{6177034713075072 	}$ \\
$H_{68}$ 	&	 0 &	 0 &	 -815/3 &	 89/6 &	 1 &	 $\frac{ -42621413028711205}{37062208278450432 	}$ \\
$H_{69}$ 	&	 4631/4 &	 -18328/3 &	 -1000/3 &	 -10/3 &	 1 &	$\frac{  -1330374174754201}{1029505785512512 	}$ \\
$H_{70}$ 	&	 -39153/20 &	 105544/15 &	 -270 &	 -77/6 &	 1 &	$\frac{  -5834645898385742195}{16472092568200192 	}$ \\
$H_{71}$ 	&	 0 &	 0 &	 0 &	 0 &	 1 &	 $\frac{ -3652233205897755459}{16472092568200192 	}$ \\
$H_{72}$ 	&	 -39153/10 &	 211088/15 &	 0 &	 0 &	 1 &	$\frac{  -40194399986166687563}{74124416556900864 	}$ \\
$H_{73}$ 	&	 17391/4 &	 114896/15 &	 0 &	 0 &	 1 &	$\frac{  -67376462435613401}{1323650295658944 }$ \\
$H_{74}$ 	&	 -3069/2 &	 -38744/3 &	 0 &	 0 &	 1 &	 -1 	\\
$H_{75}$ 	&	 -968 &	 20416 &	 0 &	 0 &	 1 &	$\frac{  -206704879201250857}{441216765219648 }$ \\
$H_{76}$ 	&	 -13706/15 &	 -92396/15 &	 0 &	 0 &	 1 &	$\frac{  -8722501888932923387}{16472092568200192 	}$ \\
$H_{77}$ 	&	 0 &	 0 &	 0 &	 0 &	 1 &	$\frac{  -703462682135213465}{1684645830838656 	}$ \\
$H_{78}$ 	&	 4631/2 &	 -36656/3 &	 0 &	 0 &	 1 &	$\frac{  -2050765293919679467}{18531104139225216 	}$ \\
$H_{79}$ 	&	 0 &	 0 &	 0 &	 0 &	 1 &	 -1 	\\
$H_{80}$ 	&	 0 &	 0 &	 0 &	 0 &	 1 &	$\frac{  -34340368851241376879}{98832555409201152 }$ \\
$H_{81}$ 	&	 -4631/15 &	 -13904/5 &	 0 &	 0 &	 1 &	 -1 \\
$H_{82}$ 	&	 0 &	 0 &	 0 &	 0 &	 1 &	 $\frac{ -10725188546965769537}{21178404730543104 	}$ \\
$H_{83}$ 	&	 0 &	 0 &	 -2810/3 &	 -39/2 &	 1 &	 -1 	\\
$H_{84}$ 	&	 0 &	 0 &	 0 &	 0 &	 1 &	$\frac{  -7417316739041385395}{18531104139225216 	}$ \\
$H_{85}$ 	&	 121/6 &	 -30595/3 &	 0 &	 0 &	 1 &	$\frac{  -21505715322664188433}{74124416556900864 	}$ \\
$H_{86}$ 	&	 0 &	 0 &	 815/3 &	 -89/6 &	 1 &	$\frac{  -10051575074463385}{18384031884152	}$ \\
$H_{87}$ 	&	 1331/4 &	 24476/3 &	 270 &	 77/6 &	 1 &	$\frac{  -2065655974432544177}{9265552069612608 	}$ \\
$H_{88}$ 	&	 -8657/30 &	 -194687/15 &	 -815/3 &	 89/6 &	 1 &	$\frac{  -10513889487465286471}{21178404730543104	}$ \\
$H_{89}$ 	&	 0 &	 0 &	 270 &	 77/6 &	 1 &	$\frac{  -102492676367157469795}{296497666227603456 	}$ \\
$H_{90}$ 	&	 4631/4 &	 -18328/3 &	 -270 &	 -77/6 &	 1 &	$\frac{  -3470421686575164043}{148248833113801728 	}$ \\
$H_{91}$ 	&	 121/3 &	 -61190/3 &	 2000/3 &	 20/3 &	 1 &	 $\frac{ -12539057139644285}{8236046284100096 }$ \\
$H_{92}$ 	&	 0 &	 0 &	 0 &	 0 &	 0 &	 0 	\\
$H_{93}$ 	&	 0 &	 0 &	 0 &	 0 &	 0 &	 0 	\\
$H_{94}$ 	&	 0 &	 0 &	 19723/20 &	 2047/20 &	 1 &	 $\frac{ -15461491234942018543}{5929953324552069120 	}$ \\
$H_{95}$ 	&	 0 &	 0 &	 0 &	 0 &	 1 &	 $\frac{ -2040849950139277}{1323650295658944 	}$ \\
$H_{96}$ 	&	 0 &	 0 &	 -1630/3 &	 89/3 &	 1 &	 $\frac{ -37632094249561791565}{148248833113801728 	}$ \\
$H_{97}$ 	&	 0 &	 0 &	 -540 &	 -77/3 &	 1 &	 $\frac{ -147229716847699567}{9265552069612608 	}$ \\
$H_{98}$ 	&	 0 &	 0 &	 0 &	 0 &	 1 &	 $\frac{ -4163309017023671941}{24708138852300288 	}$ \\
$H_{99}$ 	&	 77841/20 &	 111556/5 &	 0 &	 0 &	 1 &	 -1 	\\
$H_{100}$ 	&	 10/3 &	 105125/6 &	 0 &	 0 &	 1 &	$\frac{  -35834405989042100849}{74124416556900864 }$ \\
$H_{101}$ 	&	 0 &	 0 &	 -4000/3 &	 -40/3 &	 0 &	 0 	\\
$H_{102}$ 	&	 0 &	 0 &	 1630/3 &	 -89/3 &	 1 &	 $\frac{ -10943236189159518679}{18531104139225216 }$ \\
$H_{103}$ 	&	 0 &	 0 &	 -1630/3 &	 89/3 &	 1 &	$\frac{  -244304794290394685}{32944185136400384	}$ \\
$H_{104}$ 	&	 0 &	 0 &	 -1630/3 &	 89/3 &	 1 &	$\frac{  -6580270239524616359}{10589202365271552	}$ \\
$H_{105}$ 	&	 0 &	 0 &	 0 &	 0 &	 1 &	 $\frac{ -17483526616286112727}{24708138852300288 }$ \\
$H_{106}$ 	&	 0 &	 0 &	 -1630/3 &	 89/3 &	 1 &	 $\frac{ -2389728277891266261}{8236046284100096 	}$ \\
$H_{107}$ 	&	 107584/15 &	 20/3 &	 0 &	 0 &	 1 &	 -2 	\\
$H_{108}$ 	&	 30504/5 &	 1336/3 &	 0 &	 0 &	 1 &	 $\frac{ -66935245670393753}{661825147829472 	}$ \\
$H_{109}$ 	&	 1815/2 &	 33640/3 &	 0 &	 0 &	 1 &	 $\frac{ -10173977739002723}{55152095652456 	}$ \\
$H_{110}$ 	&	 0 &	 0 &	 0 &	 0 &	 1 &	 $\frac{ -866621514187196297}{2059011571025024 	}$ \\
$H_{111}$ 	&	 4631/4 &	 -18328/3 &	 0 &	 0 &	 1 &	$\frac{  -7493427555720786047}{26954333293418496 	}$ \\
$H_{112}$ 	&	 -1804 &	 -580/3 &	 0 &	 0 &	 1 &	$\frac{  -4771933910371470719}{9265552069612608 	}$ \\
$H_{113}$ 	&	 -34522/15 &	 316/3 &	 0 &	 0 &	 1 &	$\frac{  -56090180193586615063}{98832555409201152 }$ \\
$H_{114}$ 	&	 -3069/4 &	 -19372/3 &	 0 &	 0 &	 1 &	$\frac{  -6146435219180552237}{9265552069612608	}$ \\
$H_{115}$ 	&	 55 &	 -42050/3 &	 26569/60 &	 7921/60 &	 1 &	 $\frac{ -476307820942045182061}{1976651108184023040 	}$ \\
$H_{116}$ 	&	 0 &	 0 &	 0 &	 0 &	 1 &	 $\frac{ -19450524422641811549}{32944185136400384 }$ \\
$H_{117}$ 	&	 -93/2 &	 24215/3 &	 0 &	 0 &	 1 &	$\frac{  -8220306420511019599}{42356809461086208 	}$ \\
$H_{118}$ 	&	 -1804 &	 -580/3 &	 2187/5 &	 5929/60 &	 1 &	 $\frac{ -362958430331557939}{92655520696126080	}$ \\
$H_{119}$ 	&	 -70439/30 &	 8177 &	 0 &	 0 &	 1 &	 $\frac{ -15367554150816847711}{49416277704600576 }$ \\
$H_{120}$ 	&	 1815/2 &	 33640/3 &	 0 &	 0 &	 1 &	 $\frac{ -10173977739002723}{55152095652456 	}$ \\
$H_{121}$ 	&	 0 &	 0 &	 0 &	 0 &	 1 &	$\frac{  -1003343753899617143}{6177034713075072 	}$ \\
$H_{122}$ 	&	 4631/4 &	 -18328/3 &	 0 &	 0 &	 1 &	 $\frac{ -2082303347082636833}{9265552069612608 	}$ \\
$H_{123}$ 	&	 -328/3 &	 725/3 &	 0 &	 0 &	 1 &	$\frac{  -12006211264574547431}{24708138852300288 }$ \\
$H_{124}$ 	&	 0 &	 0 &	 0 &	 0 &	 1 &	$\frac{  -13765701958912919}{2059011571025024 	}$ \\
$H_{125}$ 	&	 -27249/10 &	 8684/15 &	 0 &	 0 &	 1 &	$\frac{  -75624575885139732659}{148248833113801728 	}$ \\
$H_{126}$ 	&	 55 &	 -42050/3 &	 0 &	 0 &	 1 &	 $\frac{ -70683455524198969843}{148248833113801728 	}$ \\
$H_{127}$ 	&	 0 &	 0 &	 0 &	 0 &	 1 &	$\frac{  -11080743495118222157}{21178404730543104 }$ \\
$H_{128}$ 	&	 170681/60 &	 103481/15 &	 0 &	 0 &	 1 &	 -1 	\\
$H_{129}$ 	&	 0 &	 0 &	 0 &	 0 &	 1 &	 $\frac{ -1157293995940733471}{4632776034806304 	}$ \\
$H_{130}$ 	&	 0 &	 0 &	 0 &	 0 &	 0 &	 0 	\\
$H_{131}$ 	&	 0 &	 0 &	 4401/5 &	 -6853/30 &	 1 &	 -1 	\\
$H_{132}$ 	&	 0 &	 0 &	 0 &	 0 &	 1 &	 $\frac{ -426427906114141689}{1029505785512512 	}$ \\
$H_{133}$ 	&	 421/6 &	 22910/3 &	 0 &	 0 &	 1 &	$\frac{  -1235497822172284283}{8984777764472832 	}$ \\
$H_{134}$ 	&	 0 &	 0 &	 4401/5 &	 -6853/30 &	 1 &	 $\frac{ -5617524380783071181}{32944185136400384 	}$ \\
$H_{135}$ 	&	 0 &	 0 &	 0 &	 0 &	 1 &	 $\frac{ -26428837039774952311}{42356809461086208 	}$ \\
$H_{136}$ 	&	 0 &	 0 &	 0 &	 0 &	 1 &	$\frac{  -73815219170205621743}{148248833113801728 	}$ \\
$H_{137}$ 	&	 0 &	 0 &	 0 &	 0 &	 1 &	$\frac{  -51576322752518046641}{296497666227603456 	}$ \\
$H_{138}$ 	&	 0 &	 0 &	 0 &	 0 &	 1 &	$\frac{  -37389454791911250173}{296497666227603456 	}$ \\
$H_{139}$ 	&	 421/6 &	 22910/3 &	 0 &	 0 &	 1 &	$\frac{  -17214181054154995319}{32944185136400384 	}$ \\
$H_{140}$ 	&	 0 &	 0 &	 0 &	 0 &	 1 &	$\frac{  -1542818265173952315}{8236046284100096 	}$ \\
$H_{141}$ 	&	 0 &	 0 &	 0 &	 0 &	 1 &	$\frac{  -1157293995940733471}{2316388017403152 	}$ \\
$H_{142}$ 	&	 20/3 &	 105125/3 &	0&	0&	1&	-1	\\
\end{longtable}
\end{center}
\end{document}